\documentclass{ebarticle}
\usepackage{ebutf8,macros}

\title{Unifying type systems for mobile processes}
\author{Emmanuel Beffara}
\date{April 20, 2015}

\DeclareUnicodeCharacter{00A4}{\vec}
\begin{document}
\maketitle

\begin{abstract}
  We present a unifying framework for type systems for process calculi.
  The core of the system provides an accurate correspondence between
  essentially functional processes and linear logic proofs; fragments of this
  system correspond to previously known connections between proofs and
  processes.
  We show how the addition of extra logical axioms can widen the class of
  typeable processes in exchange for the loss of some computational properties
  like lock-freeness or termination, allowing us to see various well studied
  systems (like i/o types, linearity, control) as instances of a general
  pattern.
  This suggests unified methods for extending existing type systems with new
  features while staying in a well structured environment and constitutes a
  step towards the study of denotational semantics of processes using
  proof-theoretical methods.
\end{abstract}

\section{Introduction}

Process calculi are a wide range of formalisms designed to model concurrent
systems and reason about them by means of term rewriting.
Their applications are diverse, from the semantics of proof systems to the
conception of concrete programming languages.
Type systems for such calculi are therefore a wide domain, with systems of
different kinds designed to capture different behaviours and ensure different
properties of processes: basic interfacing,
input-output discipline \cite{pierce-1993-typing},
linearity \cite{kobayashi-1999-linearity},
lock-freeness \cite{kobayashi-2002-type},
termination \cite{deng-2004-ensuring},
respect of communication protocols \cite{honda-1993-types,honda-2008-multiparty},
functional or sequential
behaviour \cite{berger-2001-sequentiality,honda-2004-control,yoshida-2001-strong}.

In order to better understand the diversity of calculi and uncover basic
structures and general patterns, many authors have searched for languages
with simpler or more general theory in which the most features could be
expressed by means of restrictions or codings:
asynchrony \cite{boudol-1992-asynchrony},
internal mobility \cite{sangiorgi-1996-calculus},
name fusions \cite{fu-1997-proof,parrow-1998-fusion},
solos \cite{laneve-1999-solos}, etc.
It is natural to search for similar unification in the realm of type systems,
and the aim of this paper is to make a step towards this long-term objective.
Our ideal system would be simple enough so that general properties could be
reasonably easy to obtain and expressive enough so that most interesting type
systems could naturally be expressed in it in a structured way.

For this purpose, we will take inspiration and tools in proof theory.
A useful analogue is the famous and fruitful Curry-Howard correspondence: at
the core is the simply typed λ-calculus, which matches minimal intuitionistic
logic and ensures strong normalisation.
The type language can be extended for expressiveness (with quantifiers,
dependent types, polymorphism, etc.), classical logic can be embedded in it by
CPS translation or by adding logical rules.
Furthermore, extending it with a simple type equation $D=D→D$ yields the full
untyped calculus where normalisation is lost, but the identification of this
equation leads to the definition of abstract structures that are useful for
denotational semantics.

We claim that the analogue of simple types for process calculi is to be found
in linear logic \cite{girard-1987-linear}, and we propose a new implementation
of this idea.
Of course, term assignment systems for linear logic proofs have been proposed
in the past by various
authors \cite{bellin-1994-calculus,beffara-2006-concurrent,caires-2010-session}
but no such system has yet appeared as a satisfactory type system for
processes (because of too much constraint on the syntactic structure of terms),
with the notable exception of Honda and Laurent's
result \cite{honda-2010-exact} which precisely matches a proper type system
for the π-calculus with a meaningful class of  proof nets.

The novelty of our approach is to distinguish two aspects in typing: firstly
we have a typing rule for each syntactic construct independently, secondly we
have a subtyping relation that implements logical reasoning without affecting
the structure of terms; this subtyping is nothing else than entailment in
linear logic (actually a reasonable fragment of it), which allows to use all
existing theory for reasoning about it.
In this method, we insist on treating seriously the fundamental structures of
both the process calculus and the logic, the fundamental example being that
typing is preserved both by structural congruence on processes and by logical
isomorphism between types (and these are closely related).
This is necessary for developing logic-based semantics of processes in future
works, using existing tools and methods from the semantics of proofs and of
processes.

This paper is divided in two parts.
In section \ref{sec:basic} we define our basic type system for the polyadic
π-calculus and we discuss variations around the same principles for
alternative calculi.
In section \ref{sec:existing} we review several type systems and term
assignment systems and show how they fit in our framework, by means of extra
logical axioms and syntactic restrictions.
Section \ref{sec:discussion} concludes by discussing shortcomings, extensions
and ideas for future work.

\section{The basic typed calculus}
\label{sec:basic}

\subsection{Syntax}

Processes are terms of the standard polyadic π-calculus with input-guarded
replication (and no sum in the present paper), with type annotations on name
creation.
In our type system, we will derive judgements of the form $E⊢P$ where $E$ is
an environment type and $P$ is a process term.
Such a judgement is to be understood as “$P$ is well-formed under the contract
of the environment $E$”.
Environment types are made of capability assignments of the shape $x:T$ where
$x$ is a channel name and $T$ a capability type, combined using logical
connectives.
Capability types consist of an input or output capability (written $↓$ and $↑$
respectively) together with a behaviour type for the data that is
communicated, and behaviour types are capability types combined using logical
connectives.

\begin{definition}[typed terms]
  The grammar of types and processes is defined in table \ref{table:syntax}.
\end{definition}

\begin{table}
  \begin{syntax}
  \define[Capabilities:] T, U
    \case ↓A \comment{input}
    \case ↑A \comment{output}
  \define[Behaviours:] A, B
    \case T \altcase \oc T \altcase \wn T
      \comment{capability (linear, replicable, multiple)}
    \case A ⊗ B \altcase A ⅋ B
      \comment{concatenation (independent, correlated)}
    \case 1 \altcase ⊥
      \comment{empty tuple (neutral for $⊗$ or $⅋$)}
  \define[Environments:] E, F
    \case x:T \altcase \oc x:T \altcase \wn x:T
      \comment{capability assignment}
    \case E ⊗ F \altcase E ⅋ F
      \comment{union (independent or correlated)}
    \case 1 \altcase ⊥
      \comment{empty environment (neutral for $⊗$ or $⅋$)}
  \define[Processes:] P, Q
    \case \inb{u}{¤x}.P \comment{input prefix}
    \case \inr{u}{¤x}.P \comment{replicated input prefix}
    \case \outf{u}{¤v}.P \comment{output prefix}
    \case \nop \comment{inactive process}
    \case P \para Q \comment{parallel composition}
    \case \newk{x}{A}{k}P \comment{name creation, with $k∈\{1,ω\}$}
  \end{syntax}
  \caption{The syntax of types and process terms}
  \label{table:syntax}
\end{table}

Remark that we do \emph{not} force each channel name to occur only
once in an environment type, and this is a fundamental feature of our system.
It notably allows name substitution $E[¤u/¤x]$ to make sense even when it
equalises some names.

The name creation operator $\new{x}$ is annotated with a type $A$ and a kind
$k$ that distinguishes between linear and non-linear channels.
Contrary to usual practice, the type $A$ is not the type that $x$ itself will
have, but the type of the data that $x$ will transport.
In statements where the kind and type of a channel are unimportant, we use the
standard notation $\new{x}$.

The logical connectives used in environments and behaviours are those of 
multiplicative-exponential linear logic.
This logic, recalled in table \ref{table:mell}, is used to reason about
behaviours of processes.
The key ingredient is that logical consequence is interpreted as
subtyping: if $E$ and $F$ are environments such that $E$ entails $F$, then
a process that respects $F$ will respect $E$.

\begin{definition}[subtyping]
  The subtyping preorder $≤$ over environments is such that $E≤F$ holds when
  $⊢E^⊥,F$ is provable in MELL using capability assignments as atomic
  formulas.
  The associated equivalence is written $≃$.
\end{definition}

\begin{table}
  Formulas, assuming a given set of atoms written $α,β…$:
  \begin{syntax}
  \define A, B
    \case α
    \altcase α^⊥
    \altcase A ⊗ B
    \altcase A ⅋ B
    \altcase 1
    \altcase ⊥
    \altcase \oc A
    \altcase \wn A
  \end{syntax}
  Linear negation $(⋅)^⊥$ is the involution over formulas such that:
  \begin{align*}
    (α^⊥)^⊥ &= α &
    (A ⊗ B)^⊥ &= A^⊥ ⅋ B^⊥ &
    1^⊥ &= ⊥ &
    (\oc A)^⊥ &= \wn A^⊥
  \end{align*}
  Sequents are finite multisets of formulas, they are proved using the
  following rules:
  \begin{center}
    \begin{prooftree}
      \Infer0[ax]{ ⊢ A^⊥, A }
    \end{prooftree}
  \hfil
    \begin{prooftree}
      \Hypo{ ⊢ Γ, A }
      \Hypo{ ⊢ A^⊥, Δ }
      \Infer2[cut]{ ⊢ Γ, Δ }
    \end{prooftree}
  \hfil
    \begin{prooftree}
      \Hypo{ ⊢ Γ, A }
      \Hypo{ ⊢ Δ, B }
      \Infer2[$⊗$]{ ⊢ Γ, Δ, A⊗B }
    \end{prooftree}
  \hfil
    \begin{prooftree}
      \Hypo{ &⊢ Γ, A, B }
      \Infer1[$⅋$]{ &⊢ Γ, A⅋B }
    \end{prooftree}
  \\[1ex]
    \begin{prooftree}
      \Infer0[$1$]{ ⊢ 1 }
    \end{prooftree}
  \hfil
    \begin{prooftree}
      \Hypo{ &⊢ Γ }
      \Infer1[$⊥$]{ &⊢ Γ, ⊥ }
    \end{prooftree}
  \hfil
    \begin{prooftree}
      \Hypo{ &⊢ Γ, A }
      \Infer1[$\wn$]{ &⊢ Γ, \wn A }
    \end{prooftree}
  \hfil
    \begin{prooftree}
      \Hypo{ &⊢ Γ }
      \Infer1[w]{ &⊢ Γ, \wn A }
    \end{prooftree}
  \hfil
    \begin{prooftree}
      \Hypo{ &⊢ Γ, \wn A, \wn A }
      \Infer1[c]{ &⊢ Γ, \wn A }
    \end{prooftree}
  \hfil
    \begin{prooftree}
      \Hypo{ &⊢ \wn A_1, …, \wn A_n, B }
      \Infer1[$\oc$]{ &⊢ \wn A_1, …, \wn A_n, \oc B }
    \end{prooftree}
  \end{center}
  \smallskip
  \caption{Multiplicative-exponential linear logic (MELL)}
  \label{table:mell}
\end{table}

Formulas of MELL with capability assignments as atoms and where modalities
$\wn$ and $\oc$ are only applied to literals (atoms and atom negations) will
be called \emph{environment formulas}, they will be useful for reasoning about
typed processes.
Environment types correspond to such formulas with only positive atoms, {i.e.}
without negation.

\begin{definition}[typing judgement]
  Typing judgements have the shape $E⊢P$ where $E$ is an environment type and
  $P$ is a process term.
  They are derived using the rules of table \ref{table:typing}.

  The notation $¤x:A$ where $A$ is a behaviour type stands for the
  environment type obtained by annotating each capability in $A$ by a name
  in the sequence $¤x$, respecting the left-to-right order, assuming that
  the length of $¤x$ matches the number of capabilities in $A$.
  For instance, $xyz:(T⅋1)⊗\wn U⊗⊥⊗V$ stands for $((x:T)⅋1)⊗\wn(y:U)⊗⊥⊗(z:V)$.
\end{definition}

\begin{table}
  \begin{center}
    \begin{prooftree}
      \Infer0[nop]{ ⊥ ⊢ \nop }
    \end{prooftree}
  \hfil
    \begin{prooftree}
      \Hypo{ E ⊢ P }
      \Hypo{ F ⊢ Q }
      \Infer2[para]{ E ⅋ F ⊢ P\para Q }
    \end{prooftree}
  \hfil
    \begin{prooftree}
      \Hypo{ E ≤ F }
      \Hypo{ F ⊢ P }
      \Infer2[sub]{ E ⊢ P }
    \end{prooftree}
  \\[1ex]
    \begin{prooftree}
      \Hypo{ ¤x:A ⊗ E ⊢ P }
      \Infer1[in]{ u:↓A ⊗ E ⊢ \inb{u}{¤x}.P }% ($∀i\;x_i∉E$)
    \end{prooftree}
  \hfil
    \begin{prooftree}
      \Hypo{ ¤x:A ⊗ E^! ⊢ P }
      \Infer1[in!]{% ($∀i∀j\;x_i≠y_j$)
        \wn u:↓A ⊗ E^! ⊢ \inr{u}{¤x}.P }
    \end{prooftree}
  \hfil
    \begin{prooftree}
      \Hypo{ E ⊢ P }
      \Infer1[out]{ u:↑A ⊗ (¤v:A ⅋ E) ⊢ \outf{u}{¤v}.P }
    \end{prooftree}
  \\[1ex]
    \begin{prooftree}
      \Hypo{ (x:↑A ⅋ x:↓A) ⊗ E ⊢ P }
      \Infer1[new1]{ E ⊢ \newl{x}{A}P }
    \end{prooftree}
  \hfil
    \begin{prooftree}
      \Hypo{ (\oc x:↑A ⅋ \wn x:↓A) ⊗ E ⊢ P }
      \Infer1[newω]{ E ⊢ \newr{x}{A}P }
    \end{prooftree}
  \end{center}
  \begin{itemize}
  \item In the \rulename{new} rules, the name $x$ must not
    occur in the environment $E$.
  \item In the \rulename{in} rules, the names in $¤x$ must not occur in the
    environments $E$ and $E^!$.
  \item In \rulename{in!}, $E^!$ stands for an environment of the shape
    $\oc y_1:T_1 ⊗ … ⊗ \oc y_n:T_n$.
  \end{itemize}
  \caption{Typing rules for processes}
  \label{table:typing}
\end{table}

Note that process terms have no type, in other words there is a unique type
for processes which means “well-formed”; it is also the case for instance in
i/o types with linearity, as studied in section \ref{sec:kpt}.
Of course, it would be strictly equivalent to consider that, in $E⊢P$, the
formula $E^⊥$ is the type of $P$: this is what usually happens in systems more
oriented towards logic, like those studied in sections \ref{sec:hyb}
and \ref{sec:session}.

Remark that input and output capabilities are logically not dual, in the sense
of being a negation of each other: $(u:↓A)^⊥$ and $u:↑A$ are just distinct
literals.
Actual duality between input and output is established by the typing rules for
name creation, for instance \rulename{new1} corresponds to setting the formula
$x:↑A⅋x:↓A$ as an axiom, which does represent the creation of a name $x$ with
one occurrence of each capability, where capabilities are dual.

In the statements and proofs, the types for channels in premisses of
\rulename{new} rules is used in many places.
For readability and conciseness, we introduce the following notations:
\begin{align*}
  [x]^1_A &:= x:↑A ⅋ x:↓A, &
  [x]^ω_A &:= \oc x:↑A ⅋ \wn x:↓A.
\end{align*}
In $[x]^k_A$, we may keep $k$ or $A$ implicit when the details are
unimportant.
This way, the rules \rulename{new1} and \rulename{newω} are simplified into a
single form:
\[
  \begin{prooftree}
    \Hypo{ [x]^k_A ⊗ E &⊢ P }
    \Infer1[new$k$]{ E &⊢ \newk{x}{A}{k}P }
  \end{prooftree}
  \qquad\text{where $x$ does not occur in $E$.}
\]
Moreover, in applications of the \rulename{sub} rule, we will keep the premiss $E≤F$
implicit since $E$ and $F$ are the environments of the conclusion and premiss
respectively.

\begin{lemma}\label{lemma:subtyping-commutation}
  A judgement $E⊢P$ holds if and only if it has a derivation where
  \rulename{sub} rules appear only right above \rulename{in} and
  \rulename{new} rules and at the root of the proof.
\end{lemma}
\begin{proof}
  Firstly, it is clear that successive uses of the \rulename{sub} rule can
  always be gathered into one thanks to the cut rule of MELL.
  We may thus assume without loss of generality that no \rulename{sub} rule
  occurs above another \rulename{sub} rule.
  Then one easily checks by case analysis on the proofs that each
  \rulename{sub} rule can be commuted down with any rule except \rulename{new}
  and \rulename{in} rules because these impose constraints on the context of
  their premiss.
\end{proof}

This lemma allows us to consider a restricted form of derivation when
reasoning on typed processes.
In order to establish subject reduction in the next section, we will also need
the following general properties of MELL proofs:

\begin{lemma}[substitutivity]\label{lemma:mell-subst}
  If $⊢Γ$ is a provable sequent in MELL, then for all propositional variable
  $α$ and formula $A$ the sequent $⊢Γ[A/α]$ is also provable.
\end{lemma}

\begin{lemma}[interpolation]\label{lemma:interpolation}
  Let $Γ$ and $Δ$ be two multisets of MELL formulas.
  If $⊢Γ,Δ$ is provable, then there exists a formula $F$ that contains only
  literals present in both $Γ^⊥$ and $Δ$ such that the sequents $⊢Γ,F$ and
  $⊢F^⊥,Δ$ are provable.
\end{lemma}

Both lemmas are easily proved by structural induction over proofs (a detailed
proof for lemma \ref{lemma:interpolation} can be found in
appendix \ref{app:interpolation}).
They actually hold for full linear logic but we state them in MELL because it
is the fragment we use in this paper.

\subsection{Execution}

Our presentation of execution uses structural congruence and reduction,
because it provides simpler statements than a presentation using a labelled
transition system.

\begin{definition}[structural congruence]
  The congruence $≡$ over process terms is defined by abelian monoid laws for
  parallel composition and the standard scoping rules:
  \begin{gather*}
    (P \para Q) \para R ≡ P \para (Q \para R) \qquad
    P \para Q ≡ Q \para P \qquad
    P \para 0 ≡ P \\
    \newk{x}{A}{k}\newk{y}{B}{ℓ}P ≡ \newk{y}{B}{ℓ}\newk{x}{A}{k}P \qquad
    P \para \newk{x}{A}{k}Q ≡ \newk{x}{A}{k}(P \para Q)
  \end{gather*}
  where $x≠y$ and $x$ does not occur free in $P$ in the last rule.
\end{definition}

\begin{lemma}\label{lemma:congruence}
  Typing is preserved by structural congruence.
\end{lemma}
\begin{proof}
  This is proved by checking each axiom of structural congruence.
  Most cases are direct, the only technical point is the proof that if
  $\newk{x}{A}{k}(P\para Q)$ is typeable and $x$ does not occur in $P$, then
  $P\para\newk{x}{A}{k}Q$ has the same type.
  We transform a generic typing of $\newk{x}{A}{k}(P\para Q)$ into a typing of
  $P\para\newk{x}{A}{k}Q$ as follows:
  \[
    \begin{prooftree}
      \Hypo{ E ⊢ P }
      \Hypo{ F ⊢ Q }
      \Infer2[para]{ E ⅋ F ⊢ P \para Q }
      \Infer1[sub]{ [x]^k_A ⊗ G &⊢ P \para Q }
      \Infer1[new$k$]{ G &⊢ \newk{x}{A}{k}(P \para Q) }
    \end{prooftree}
    \quad\to\quad
    \begin{prooftree}
      \Hypo{ E ⊢ P }
      \Hypo{ F &⊢ Q }
      \Infer1[sub]{ [x]^k_A ⊗ H &⊢ Q }
      \Infer1[new$k$]{ H &⊢ \newk{x}{A}{k}Q }
      \Infer2[para]{ E ⅋ H ⊢ P \para \newk{x}{A}{k}Q }
      \Infer1[sub]{ G ⊢ P \para \newk{x}{A}{k}Q }
    \end{prooftree}
  \]
  In order to find the environment $H$, remark that the subtyping on the left
  is justified by an MELL proof of $⊢([x]^k_A)^⊥,G^⊥,E,F$.
  By lemma \ref{lemma:interpolation}, there is a formula $H$ such that
  $⊢G^⊥,E,H$ and $⊢H^⊥,([x]^k_A)^⊥,F$ are provable and the literals in $H$ occur
  both in $G,E^⊥$ and in $([x]^k_A)^⊥,F$.
  By hypothesis $x$ does not occur in $P$ hence not in $E$, and not in $G$
  either by the side-condition on \rulename{new$k$}, so there is no $x$ in $H$.
  Hence the literals in $H$ occur in $F$ so they are positive and $H$ is an
  environment type.
  The proofs of $⊢G^⊥,E,H$ and $⊢H^⊥,([x]^k_A)^⊥,F$ justify the \rulename{sub}
  rules on the right.
  The other cases are detailed in appendix \ref{app:congruence}.
\end{proof}

\begin{definition}[reduction]
  Reduction is the relation $\redpi{ℓ}$ where $ℓ$ is either a name or the
  symbol $τ$.
  It is generated by the rules
  \begin{align*}
    \outf{u}{¤v}.P \para \inb{u}{¤x}.Q
      &\redpi{u} P \para Q[¤v/¤x] &
    \outf{u}{¤v}.P \para \inr{u}{¤x}.Q
      &\redpi{u} P \para Q[¤v/¤x] \para \inr{u}{¤x}.Q
  \end{align*}
  extended to arbitrary contexts as
  \[
    \begin{prooftree}
      \Hypo{ P \redpi{ℓ} P' }
      \Infer1{ P\para Q \redpi{ℓ} P'\para Q }
    \end{prooftree}
    \;
    \begin{prooftree}
      \Hypo{ P \redpi{ℓ} P' }
      \Hypo{ ℓ ≠ u }
      \Infer2{ \new{u}P \redpi{ℓ} \new{u}P' }
    \end{prooftree}
    \;
    \begin{prooftree}
      \Hypo{ P \redpi{u} P' }
      \Infer1{ \newl{u}{A}P \redpi{τ} P' }
    \end{prooftree}
    \;
    \begin{prooftree}
      \Hypo{ P \redpi{u} P' }
      \Infer1{ \newr{u}{A}P \redpi{τ} \newr{u}{A}P' }
    \end{prooftree}
  \]
  and saturated under structural congruence.
\end{definition}

The only difference with standard reduction is that we delete linear name
creations as soon as their name is used.
This is consistent with the linearity requirement, moreover in typed processes
this requirement is fulfilled.
In plain π-calculus this operation would be handled by the congruence rule
$\new{x}P≡P$ if $x$ is not free in $P$, but we choose not to use this approach
here in order to avoid an extra kind of “new” operator just for this case.

\begin{theorem}[subject reduction]\label{thm:subject-reduction}
  For all typed term $Γ⊢P$ and execution step $P\redpi{τ}P'$,
  the judgement $Γ⊢P'$ is derivable.
\end{theorem}
\begin{proof}
  Thanks to lemma \ref{lemma:congruence}, we can reason up to structural
  congruence.
  For an interaction step on a linear channel, we have
  $\newl{u}{A}(\outf{u}{¤v}.P\para\inb{u}{¤x}.Q)\redpi{τ}P\para Q[¤v/¤x]$.
  The left-hand side is typed as follows (using the simplification from
  lemma \ref{lemma:subtyping-commutation}):
  \begin{prooftree*}
    \Hypo{ E &⊢ P }
    \Infer1[out]{ u:↑A ⊗ (¤v:A ⅋ E) &⊢ \outf{u}{¤v}.P }
    \Hypo{ ¤x:A ⊗ F &⊢ Q }
    \Infer1[in]{ u:↓A ⊗ F &⊢ \inb{u}{¤x}.Q }
    \Infer2[para]{ (u:↑A⊗(¤v:A⅋E))⅋(u:↓A⊗F) ⊢ \outf{u}{¤v}.P\para\inb{u}{¤x}.Q }
    \Infer1[sub]{ (u:↑A⅋u:↓A) ⊗ H ⊢ \outf{u}{¤v}.P\para\inb{u}{¤x}.Q }
    \Infer1[new1]{ H ⊢ \newl{u}{A}(\outf{u}{¤v}.P\para\inb{u}{¤x}.Q) }
  \end{prooftree*}
  with the hypothesis that no name in $¤x$ occurs in $F$.
  The \rulename{sub} rule is justified by an MELL proof of
  $⊢(((u:↑A)^⊥⊗(u:↓A)^⊥)⅋H^⊥),(u:↑A⊗(¤v:A⅋E))⅋(u:↓A⊗F)$.
  By lemma \ref{lemma:mell-subst}, we can replace the atomic formula $u:↓A$ by
  $¤v:A$ and the atomic formula $u:↑A$ by $(¤v:A)^⊥$, then we get a proof of
  \[
    ⊢ (¤v:A ⊗ (¤v:A)^⊥) ⅋ H^⊥, ((¤v:A)^⊥ ⊗ (¤v:A ⅋ E)) ⅋ (¤v:A ⊗ F)
  \]
  The following sequents are easily proved in MELL:
  \begin{gather*}
    ⊢ H^⊥, (¤v:A⅋(¤v:A)^⊥)⊗H \\
    ⊢ (¤v:A ⅋ ((¤v:A)^⊥ ⊗ E^⊥)) ⊗ ((¤v:A)^⊥ ⅋ F^⊥), E ⅋ (¤v:A ⊗ F)
  \end{gather*}
  so we get a proof of $⊢H^⊥,(¤v:A⊗F)⅋E$ by composition, which justifies the
  typing:
  \begin{prooftree*}
    \Hypo{ E ⊢ P }
    \Hypo{ ¤v:A ⊗ F ⊢ Q[¤v/¤x] }
    \Infer2[para]{ E ⅋ (¤v:A ⊗ F) &⊢ P\para Q[¤v/¤x] }
    \Infer1[sub]{ H &⊢ P\para Q[¤v/¤x] }
  \end{prooftree*}
  The case of a reduction on a non-linear channel is similar, with some extra
  work to handle duplication; details can be found in
  appendix \ref{app:subject-reduction}.
\end{proof}

Remark that the introduction of negated atoms in the proof above makes us go
through environment \emph{formulas} that are not proper types, although
composition by cut provides a subtyping between environment types.
These intermediate steps correspond to the introduction in our system of axiom
rules that transport arbitrary behaviours (here the $¤v:A$) with no
counterpart in the terms, as a decomposition of the name passing mechanism.
This is similar to the central role of axioms in the
proofs-as-schedules \cite{beffara-2014-proof} paradigm.

\subsection{The role of “new”}

The subject reduction property is formulated for reductions on
private channels, {i.e.} names that are explicitly created in the term.
Indeed, the property fails without this assumption: not only is the type not
preserved (which is expected in the case of linear capabilities), but
communicated data may not have proper types.
For instance, in a typed term like
\begin{prooftree*}
  \Hypo{ E ⊢ P }
  \Infer1[out]{ u:↑A ⊗ (v:A ⅋ E) ⊢ \outf{u}{v}.P }
  \Hypo{ x:B ⊗ F ⊢ Q }
  \Infer1[in]{ u:↓B ⊗ F ⊢ \inb{u}{x}.Q }
  \Infer2[para]{ (u:↑A ⊗ (v:A ⅋ E)) ⅋ (u:↓B ⊗ F)
    ⊢ \outf{u}{v}.P \para \inb{u}{x}.Q }
\end{prooftree*}
the name $v$ has type $A$ but the name $x$ has type $B$, and there is no
reason that $A$ and $B$ are compatible, thus in general we cannot type the
reduct $P\para Q[v/x]$.

We do not consider this a serious defect of the system, it is mostly a matter
of presentation.
Indeed, the purpose of typing is to ensure proper composition of processes,
and the creation of channels is part of the composition operation.
Therefore, composition only makes sense in the presence of name creation
operators, and in the example above neither \rulename{new1} nor
\rulename{newω} applies if $A$ and $B$ do not match.
We could reformulate our system so that situations like the one above are
forbidden by typing.
A natural approach would be to enforce syntactic constraints on environment
types, for instance that linear capability assignments occur at most once,
that dual capability assignments have matching types, etc.
We chose not to include such restrictions for simplicity, relying on the
above justification.

\subsection{Properties of typed processes}

It can be proved that processes typed in our system are well-behaved:
\begin{theorem}[termination]
  Typed processes have no infinite sequence of transitions on
  private names.
\end{theorem}
\begin{theorem}[lock-freeness]
  In every execution of a typed closed term, every active output eventually
  interacts with an input.
\end{theorem}
Proofs of these facts can be obtained by realisability techniques, as in
previous work by the
author \cite{beffara-2005-logique,beffara-2006-concurrent}, or by syntactic
means by relating execution with the cut-elimination procedure of linear
logic.
We do not include proofs because they are out of the scope of this paper.
Nevertheless, a fundamental point in the arguments is that they rely on the
consistency of linear logic (through the cut-elimination property).
In relaxations of the system introduced in section \ref{sec:existing}, we will
express systems which do not enjoy those properties, by means of inconsistent
extensions of this logic.

\subsection{Variations}

The choice of the polyadic π-calculus in the presentation of our system is
justified by the fact that it is very expressive and also very standard.
However, we can adapt our approach to most variants of the calculus.

\subparagraph*{Asynchrony}

This is the restriction on outputs to have no
continuations \cite{boudol-1992-asynchrony}.
The typing of a free output atom $\outf{u}{¤v}$,
considered as a simple process $\outf{u}{¤v}.\nop$, is as follows:
\[
  \begin{prooftree}
    \Infer0[nop]{ ⊥ ⊢ \nop }
    \Infer1[out]{ u:↑A ⊗ (¤v:A ⅋ ⊥) ⊢ \outf{u}{¤v}.\nop }
  \end{prooftree}
  \quad\leadsto\quad
  \begin{prooftree}
    \Infer0[out-async]{ u:↑A ⊗ ¤v:A ⊢ \outf{u}{¤v} }
  \end{prooftree}
\]
where the simplified type is appropriate since it is linearly equivalent to
the one on the left, because of neutrality of $⊥$ for $⅋$.
Apart from this rule, nothing is changed in the system for the asynchronous
π-calculus.

\subparagraph*{Internal mobility}

This is the restriction that output prefixes only communicate distinct bound
names \cite{sangiorgi-1996-calculus}.
This simplifies the theory of the calculus and makes it symmetric like CCS.
In our type system, we also get the symmetry in typing rules.
For this purpose we can introduce duality over behaviour types:
\begin{definition}[duality]
  For a behaviour type $A$, the dual $\dual{A}$ is defined inductively as:
  \begin{align*}
    \dual{↑A} &:= ↓A &
    \dual{\oc A} &:= \wn\dual{A} &
    \dual{A ⊗ B} &:= \dual{A} ⅋ \dual{B} &
    \dual{1} &:= ⊥ \\
    \dual{↓A} &:= ↑A &
    \dual{\wn A} &:= \oc\dual{A} &
    \dual{A ⅋ B} &:= \dual{A} ⊗ \dual{B} &
    \dual{⊥} &:= 1
  \end{align*}
\end{definition}

The dual $\dual{A}$ of a formula $A$ is a form of linear negation, except that
the dual of a capability $↑A$ is the capability $↓A$, whereas negations keep
capabilities unaffected in our environment formulas.
Note that we do not apply duality inside the capability, since we follow the
approach of i/o types, where this convention is the appropriate one.
Nevertheless, logically, the output capability contains a negation, as
illustrated by the bound output rule below.

\begin{lemma}[generalised new]\label{lemma:new*}
  The following rule is derivable, assuming the tuple $¤x$ is made of
  pairwise distinct names:
  \begin{prooftree*}
    \Hypo{ (¤x:A ⅋ ¤x:\dual{A}) ⊗ E ⊢ P }
    \Infer1[new*]{ E ⊢ \new{¤x}P }
  \end{prooftree*}
\end{lemma}
\begin{proof}
  This is proved by induction $A$.
  The base case is when $A$ is a linear or exponential capability, then one of
  the \rulename{new} rules applies directly.
  If $A=⊥$, then $¤x$ is empty and we have
  $(¤x:A⅋¤x:\dual{A})⊗E=(⊥⅋1)⊗E≃1⊗E≃E$ so the rule holds by linear
  equivalence.
  The case $A=1$ is similar.
  If $A=B⅋C$ then $¤x$ splits as $¤y,¤z$ so that we have
  \begin{prooftree*}
    \Hypo{ ((¤y:B ⅋ ¤z:C) ⅋ (¤y:\dual{B} ⊗ ¤z:\dual{C})) ⊗ E ⊢ P }
    \Infer1[sub]{ (¤z:C ⅋ ¤z:\dual{C}) ⊗ (¤y:B ⅋ ¤y:\dual{B}) ⊗ E ⊢ P }
    \Infer1[new*]{ (¤y:B ⅋ ¤y:\dual{B}) ⊗ E ⊢ \new{¤z}P }
    \Infer1[new*]{ E ⊢ \new{¤y}\new{¤z}P }
  \end{prooftree*}
  where the \rulename{sub} rule is justified by a simple MLL proof.
  The case $A=B⊗C$ is similar.
\end{proof}

Using this lemma, we can derive a typing rule for bound output:
\[
  \begin{prooftree}
    \Hypo{ ¤x:\dual{A} ⊗ E ⊢ P }
    \Infer1[out]{ u:↑A ⊗ (¤x:A ⅋ (¤x:\dual{A} ⊗ E)) ⊢ \outf{u}{¤x}.P }
    \Infer1[sub]{ u:↑A ⊗ (¤x:A ⅋ ¤x:\dual{A}) ⊗ E ⊢ \outf{u}{¤x}.P }
    \Infer1[new*]{ u:↑A ⊗ E ⊢ \new{¤x}\outf{u}{¤x}.P }
  \end{prooftree}
  \quad\leadsto\quad
  \begin{prooftree}
    \Hypo{ ¤x:\dual{A} ⊗ E ⊢ P }
    \Infer1[out-bound]{ u:↑A ⊗ E ⊢ \outb{u}{¤x}.P }
  \end{prooftree}
\]

\subparagraph*{Fusions}

Our system can be extended to calculi with free input, such as
the fusion calculus \cite{parrow-1998-fusion}.
The appropriate formulation is with a preorder over
names \cite{hirschkoff-2013-name} generated by “arcs” $a/b$ which are explicit
substitution atoms.
The logical meaning of an arc is an implication $\oc(a:T⊸b:T)$ for any
capability type $T$: it allows a capability on $a$ to be used as a
capability on $b$; the modality is because the substitution is permanently
available.
The typing rule would be an axiom like $\wn(a:T⊗(b:T)^⊥)⊢a/b$; this
implies the handling of negative atoms, which may have an impact on the
structure of the system.
We defer the formal development of this extension to future work, since it
exceeds the scope of the present paper.

\section{Existing systems as fragments and extensions}
\label{sec:existing}

In this section, we describe formally how our system can express known type
systems for processes, using relaxations and identifying fragments.
By \emph{relaxation}, we mean that we add new logical rules to MELL in order
to prove more subtypings.
The resulting system need not be logically consistent, the minimal
requirements are that the new rules preserve
\begin{itemize}
\item the interpolation lemma, so that typing is still preserved under
  structural congruence,
\item the substitution lemma, so that subject reduction still holds.
\end{itemize}

\subsection{Linearity and i/o types}
\label{sec:kpt}

We show here how our system can express plain i/o types à la Pierce and
Sangiorgi \cite{pierce-1993-typing} and their extension with linearity by
Kobayashi, Pierce and Turner \cite{kobayashi-1999-linearity} (hereafter
referred to as KPT).
We develop the relationship only with KPT, since plain i/o types are its
fragment without linear types.
We refer the reader to the original paper for the notations.

\begin{definition}
  Let $ℒ$ be the fragment of KPT where
  \begin{itemize}
  \item in channel types, only pure input or output capabilities are used,
  \item linear channel creations must create both capabilities,
  \item the boolean data type is not used.
  \end{itemize}
  The translation $⟦⋅⟧$ maps channel types of $ℒ$ to channel types, tuples of
  channel types of $ℒ$ to behaviour types and contexts of $ℒ$ to environment
  types as follows:
  \begin{gather*}
    ⟦ \oc^1¤T ⟧ := ↑⟦¤T⟧ \quad
    ⟦ \wn^1¤T ⟧ := ↓⟦¤T⟧ \quad
    ⟦ \oc^ω¤T ⟧ := \oc↑⟦¤T⟧ \quad
    ⟦ \wn^ω¤T ⟧ := \wn↓⟦¤T⟧ \\
    ⟦ T_1…T_n ⟧ := ⟦T_1⟧ ⊗ … ⊗ ⟦T_n⟧ \\
    ⟦ x:\oc^m¤T ⟧ := x:⟦\oc^m¤T⟧ \quad
    ⟦ x:\wn^m¤T ⟧ := x:⟦\wn^m¤T⟧ \quad
    ⟦ x:↕^m¤T ⟧ := x:⟦\oc^m¤T⟧ ⅋ x:⟦\wn^m¤T⟧ \\
    ⟦ x_1:T_1, … x_n:T_n ⟧ := ⟦x_1:T_1⟧ ⊗ … ⊗ ⟦x_n:T_n⟧
  \end{gather*}
\end{definition}

The restriction on channel types is of minor importance as it can be lifted by
a simple coding: communicating channels with no capabilities is useless so it
can be removed, and instead of communicating a channel with both capabilities,
one can communicate each capability as distinct arguments.
As for the restriction on channel creation, it is harmless since a channel
created without both capabilities will never have any communication.
The exclusion of booleans is simply because our system, for simplicity of
presentation, does not include base data types; extension of the system with
such types is not problematic.

\begin{theorem}
  A typing judgement $Γ⊢P$ holds in $ℒ$ if and only if the judgement $⟦Γ⟧⊢P$
  holds in our system extended with the logical equivalences
  \begin{align*}
    A ⊗ B &≃ A ⅋ B &
    1 &≃ ⊥ &
    \oc A &≃ \wn A
  \end{align*}
\end{theorem}
\begin{proof}[Sketch of proof]
  Using these equivalences, environment types, up to associativity,
  commutativity and neutrality, are just multisets of capability assignments of
  the shape $x:T$ or $\oc x:T$.
  Moreover, the multiplicity of each $\oc x:T$ does not matter.
  Similarly, behaviour types are now just tuples of capabilities.
  This provides a reverse mapping from our types to those of $ℒ$.
  Then it is easy to check that each typing rule in $ℒ$ can be derived in our
  system, which proves the direct implication.
  For the reverse implication, we just have to check that our rules are also
  valid in $ℒ$, only taking care of multiple occurrences of a name in an
  environment type by appropriate constraints on the use of contraction.
\end{proof}

The addition of the logical equivalences can be achieved by adding to the
proof rules of MELL any rules that implement these equations as linear
equivalences (as new axiom rules or as new introduction rules for the
connectives involved; these methods are equivalent).
It is not hard to check that this relaxation does preserve the interpolation
and substitution lemmas.
Of course, lock-freeness and termination are lost, and this is directly
related to the fact that the equivalences make the logic inconsistent: cut
elimination is lost.

\subsection{Control, sequentiality, etc.}
\label{sec:hyb}

In a series of
works \cite{berger-2001-sequentiality,berger-2005-genericity,honda-2004-control,yoshida-2001-strong},
Berger, Honda and Yoshida studied refinements of i/o types with linearity
where various properties are enforced including sequentiality, strong
normalisation, or the behaviour of functional computation with control.
The latter system (hereafter called HYB, we refer the reader to the
paper \cite{honda-2004-control} for the notations) was put in precise
correspondence with proof nets for polarised linear logic by Honda and
Laurent \cite{honda-2010-exact} and this correspondence fits in our system.
\begin{definition}
  The translation $⟦⋅⟧$ from HYB types to behaviour types, HYB contexts to
  environment types and processes to processes is defined as follows:
  \begin{gather*}
    ⟦ (¤{τ})^? ⟧ := ↑\dual{⟦¤{τ}⟧}, \quad
    ⟦ (¤{τ})^! ⟧ := ↓⟦¤{τ}⟧, \quad
    ⟦ τ_1 … τ_n ⟧ :=
      (\oc⟦τ_1⟧ ⅋ \wn\dual{⟦τ_1⟧}) ⊗ … ⊗ (\oc⟦τ_n⟧ ⅋ \wn\dual{⟦τ_n⟧}) \\
    ⟦ x:τ_O ⟧ := \oc x:⟦τ_O⟧, \quad
    ⟦ x:τ_I ⟧ := \oc x:\dual{⟦τ_I⟧} ⅋ \wn x:⟦τ_I⟧, \\
    ⟦ \inr{x}{y_1…y_n}.P ⟧ := \inr{x}{y_1y'_1…y_ny'_n}.⟦ P ⟧ 
      \qquad\text{with $y'_1,…,y'_n$ fresh,} \\
    ⟦ \outb{x}{y_1…y_n}P ⟧ :=
    \new{y_1…y_n}(\outf{x}{y_1y_1…y_ny_n} \para ⟦ P ⟧) .
  \end{gather*}
\end{definition}

This translation is essentially the isomorphism between π-calculus types à la
HYB and formulas of LLP, plus the capability indications.
A crucial difference is that we have to code every communication of a single
name as the communication of a pair for the input and the output capabilities,
since in HYB an input type $(¤{τ})^!$ actually allows the presence of outputs,
while our type system does not allow sending both capabilities as a single
argument.
Through this translation, we do capture HYB's typing, and the following
theorem is proved by writing translations between the two systems:

\begin{theorem}
  A judgement $⊢P\triangleright x_1:τ_1,…,x_n:τ_n$ is derivable in HYB if and
  only if the judgement $⟦x_1:τ_1⟧⊗…⊗⟦x_n:τ_n⟧⊢P$ is derivable in our system
  extended with the equations $A⊗B≃A⅋B$ and $1≃⊥$.
\end{theorem}

Again, the identification of dual connectives makes the underlying logic
degenerate, and indeed the logic above does not ensure normalisation.
Honda and Laurent enumerate several restrictions of this system: acyclicity of
name dependence, input or output determinism, etc; in our system, these
restrictions mean than we do not identify dual connectives, then the theorem
above extends as an embedding of LLP/$π^c$ into our system.

The same approach can be used to handle other type systems of the same family,
we leave the formalisation of the correspondence for those systems to future
work.

\subsection{Session types}
\label{sec:session}

Caires and Pfenning \cite{caires-2010-session} formulated an equivalence
between dyadic session types \cite{honda-1993-types} and intuitionistic linear
logic, using a suitable interpretation of the connectives: $u:A⊸B$ means “on
$u$, receive a channel of type $A$ then proceed according to $B$”, dually
$u:A⊗B$ means “send a channel of type $A$, then proceed according to $B$”.
This implies that the type of a channel must change during an interaction,
following the progress of the session.
This seems to be incompatible with type systems in which a type is assigned to
each channel in a static way, including the present work, however the same
authors with DeYoung and Toninho \cite{deyoung-2012-cut} found a
reformulation of their correspondence (hereafter called DCPT) in the
asynchronous π-calculus where this contradiction vanishes.
The trick is that these channels must never have more than one active
occurrence per polarity and this can be turned into linearity by applying to
synchronous processes a translation $⟦⋅⟧$ defined as follows:
\begin{align*}
  ⟦ \inb{u}{x}.P ⟧ &:= \inb{u}{xu'}.⟦P⟧[u'/u] &
  ⟦ \outf{u}{v}.P ⟧ &:= \new{u'}(\outf{u}{vu'} \para ⟦P⟧[u'/v])
\end{align*}
where $u'$ is a fresh name that represents the state of $u$ at the next step
of interaction.
Of course, this translation does not make sense for general processes, but in
the case of the interaction discipline enforced by session types, this
transformation is perfectly adequate.

\begin{theorem}
  Let $⟦⋅⟧$ be the following translation from LL formulas to channel types:
  \begin{gather*}
    ⟦ 1 ⟧ := ↑⊥ \qquad
    ⟦ A ⊗ B ⟧ := ↑(⟦A^⊥⟧ ⅋ ⟦B^⊥⟧) \qquad
    ⟦ \oc A ⟧ := ↑\wn⟦A⟧ \\
    ⟦ ⊥ ⟧ := ↓⊥ \qquad
    ⟦ A ⅋ B ⟧ := ↓(⟦A⟧ ⅋ ⟦B⟧) \qquad
    ⟦ \wn A ⟧ := ↓\wn⟦A⟧ \\
    ⟦ x_1:A_1, …, x_n:A_n ⟧ := x_1:⟦A_1⟧ ⊗ … ⊗ x_n:⟦A_n⟧
  \end{gather*}
  If $Γ⊢P::x:A$ is derivable in DCPT then $⟦Γ⟧⊗x:⟦A^⊥⟧⊢P$ holds in our
  system. \\
  If $⟦Γ⟧⊗x:⟦A^⊥⟧⊢P$ holds, then $Γ⊢P'::x:A$ is derivable in DCPT for some
  $P'≡P$.
\end{theorem}
\begin{proof}
  The direct implication simply consists in checking that each rule of DCPT
  translates in our system, which is straightforward.
  For the reverse implication, we establish a standardisation result for our
  type system (applied to the π-calculus with internal mobility) which
  essentially eliminates the \rulename{sub} rule by cut elimination; we just
  have to check that all permutations involved are structural congruences.
\end{proof}

\section{Discussion}
\label{sec:discussion}

\subparagraph*{More systems}

Our results are formulated in a π-calculus without choice using MELL as a
subtyping logic.
We chose to present this system since it illustrates the fundamental ideas of
our approach, but it can be naturally extended to a type system for the
π-calculus with choice, more liberal replication,
genericity \cite{berger-2005-genericity} etc using full linear logic, with
additives and second-order quantification.

We also conjecture that it should be possible to embed
systems of a different kind using modalities different from the $\oc$ and
$\wn$ of linear logic.
In particular, type systems that ensure termination by stratification of
names \cite{deng-2004-ensuring} should correspond to using our basic system
but replacing MELL with a form of light logic \cite{girard-1998-light} where
the operations on exponentials are constrained using stratification techniques
that are (at least superficially) similar.

\subparagraph*{Synchrony, or lack thereof}

The lock-freeness property that the system ensures is important but it implies
a serious defect of our system: it is very weak at dealing with prefixing.
A witness of this fact can be seen in the following derivation:
\[
  \begin{prooftree}
    \Hypo{ y:B ⊗ E ⊢ P }
    \Infer1[in]{ v:↓B ⊗ E ⊢ \inb{v}{y}.P }
    \Infer1[sub]{ x:A ⊗ F ⊢ \inb{v}{y}.P }
    \Infer1[in]{ u(x):↓A ⊗ F ⊢ \inb{u}{x}.\inb{v}{y}.P }
  \end{prooftree}
\]
Assuming that the names $u,v,x,y$ are all distinct, it is easy to prove (by
reasoning on the MELL proof of $⊢(x:A)^⊥⅋F^⊥,v:↓B⊗E$) that $F$ can actually be
written $v:↓B⊗F'$ up to associativity and commutativity, and that subsequently
the subtypings $x:A⊗F'≤E$ and hence $x:A⊗y:B⊗F'≤y:B⊗E$ hold.
Therefore the term $\inb{v}{y}.\inb{u}{x}.P$ will also by typeable by the same
type as above.
Hence our types are preserved by the equivalence
\[
  \inb{u}{x}.\inb{v}{y}.P ≃ \inb{v}{y}.\inb{u}{x}.P
\]
The same argument applies to output prefixes and commutation between inputs
and outputs.
A consequence of this observation is that any typed equivalence over processes
must include the rule above, in other words our type system actually tells
about a very asynchronous calculus (this is nearly the calculus of
solos \cite{laneve-1999-solos} with restrictions on scopes, except that
prefixes can freely commute but not interact).

A deep reason for this state of things is that the discipline on names in
process composition stems from proof composition in linear logic, which
fundamentally works by enforcing acyclicity and connectedness in connections
between proofs \cite{danos-1989-structure}, in a \emph{commutative} context.
Indeed, the multiplicative connectives can be interpreted as follows:
\begin{itemize}
\item $E⊗F⊢P$ means that $P$ is expected to behave well in an environment that
  provides some behaviour for $E$ and some behaviour for $F$, and those are
  \emph{independent}.
\item $E⅋F⊢P$ means that $P$ is expected to behave well when these two
  behaviours are \emph{correlated}, {i.e.} some events in $E$ can be prefixed
  by events in $F$ and vice-versa.
\end{itemize}
With only this kind of information, there is no hope to have a type system
that would accept $\inb{a}{}.\inb{b}{}\para\outf{a}{}.\outf{b}{}$ but would
reject $\inb{a}{}.\inb{b}{}\para\outf{b}{}.\outf{a}{}$.
The only way out of this problem is either to extend the logic with
non-commutative connectives, or to introduce other forms of dependencies, for
instance through quantification.

\subparagraph*{Semantics}

This paper does not discuss semantic aspects of logic and processes, however
these are fundamental motivations of our approach.
We claim that the method of starting with a very constrained system and the
relaxing it in a controlled way using logical axioms should be fruitful
in this respect.

Realisability can be used to extract interpretations of formulas and terms
from syntax itself, using orthogonality as a generic form of testing.
It is efficient, in particular, for specifying operational properties of
processes, among which termination and lock-freeness.
Capabilities get interpreted by basic operational definitions while logic is
interpreted as in phase semantics, which justifies our use of entailment as
subtyping since, in such semantics, $E⊢F$ does imply the inclusion of $E$ into
$F$.
Besides, consistency of phase interpretation accepts some axioms (like the mix
rule or arbitrary weakening) but not others, which justifies the effects of
adding those axioms in our subtyping logic.

Another promising direction is the use of denotational semantics of proofs as
a way to build semantics of processes.
Evidence for this can be found, for instance, in the relational model of
linear logic: it is a non-trivial model of proofs, yet it supports the
identification of opposite types, as used in section \ref{sec:kpt} to rebuild
i/o types.
Using an appropriate interpretation for capability types, this should provide
meaningful denotational models for i/o-typed processes.
Besides, the flexibility of the relational model makes it suitable to
interpret differential linear logic, in which it is possible to formulate
encodings of processes of the calculus of solos \cite{ehrhard-2010-acyclic}.
Our approach thus provides new tools for the study of denotational models of
processes.
This could for instance extend a line of work of Varacca and
Yoshida \cite{varacca-2006-typed} interpreting the π-calculus in event
structures using logical constructs.

\bibliography{biblio}

\appendix
\section{Technical appendix}

\subsection{Interpolation lemma (lemma \ref{lemma:interpolation})}
\label{app:interpolation}

\begin{proof}
  We reason by induction on a cut-free proof $π$ of $⊢Γ,Δ$.
  \begin{itemize}
  \item If $π$ is an axiom rule, then three cases may occur:
    \begin{itemize}
    \item either $Γ$ and $Δ$ are equal and are a single formula, then $F:=Δ$
      works,
    \item or $Γ=A^⊥,A$ for some $A$ and $Δ$ is empty, then $F:=⊥$ works,
    \item or $Γ$ is empty and $Δ=A^⊥,A$ for some $A$, then $F:=1$ works.
    \end{itemize}
  \item If $π$ is a $1$ rule then either $Γ=∅$ and $Δ=1$ or $Γ=1$ and $Δ=∅$.
    In either case, $Γ^⊥,Δ$ is a singleton $\{F\}$ where $F$ provides the
    expected result.
  \item It $π$ ends with a $⊥$ rule, it has the shape
    \[
      \begin{prooftree}
        \Hypo{ &⊢ Γ', Δ' }
        \Infer1[$⊥$]{ &⊢ Γ', Δ', ⊥ }
      \end{prooftree}
      \qquad\text{with}\qquad
      \begin{aligned}
        & Γ = Γ', \quad Δ = ⊥, Δ' &&\text{or} \\
        & Γ = Γ', ⊥, \quad Δ = Δ'
      \end{aligned}
    \]
    We can apply the induction hypothesis on $⊢Γ',Δ'$, yielding proofs of
    $⊢Γ',F$ and $⊢F^⊥,Δ'$, and conclude by adding a $⊥$ rule on the
    appropriate side.
  \item If $π$ ends with a $⊗$ rule, it has the shape
    \[
      \begin{prooftree}
        \Hypo{ ⊢ Γ_1, A_1, Δ_1 }
        \Hypo{ ⊢ Γ_2, A_2, Δ_2 }
        \Infer2[$⊗$]{ ⊢ Γ_1, Γ_2, A_1⊗A_2, Δ_1, Δ_2 }
      \end{prooftree}
      \qquad\text{with}\qquad
      \begin{aligned}
        & Γ = Γ_1, Γ_2, \quad Δ = A_1⊗A_2, Δ_1, Δ_2 &&\text{or} \\
        & Γ = Γ_1, Γ_2, A_1⊗A_2, \quad Δ = Δ_1, Δ_2
      \end{aligned}
    \]
    In the first case, we proceed as follows using the induction hypothesis on
    $Γ_i$ and $Δ_i,A_i$ for each $i$:
    \[
      \begin{prooftree}
        \Infer0[IH]{ ⊢ Γ_1, F_1 }
        \Infer0[IH]{ ⊢ Γ_2, F_2 }
        \Infer2[$⊗$]{ ⊢ Γ_1, Γ_2, F_1⊗F_2 }
      \end{prooftree}
      \qquad
      \begin{prooftree}
        \Infer0[IH]{ ⊢ F_1^⊥, A_1, Δ_1 }
        \Infer0[IH]{ ⊢ F_2^⊥, A_2, Δ_2 }
        \Infer2[$⊗$]{ ⊢ F_1^⊥, F_2^⊥, A_1⊗A_2, Δ_1, Δ_2 }
        \Infer1[$⅋$]{ ⊢ F_1^⊥⅋F_2^⊥, A_1⊗A_2, Δ_1, Δ_2 }
      \end{prooftree}
    \]
    These proofs provide the expected conclusions, with $F:=F_1⊗F_2$.
    As for the constraints on literals, the induction hypothesis gives us that
    the atoms in each $F_i$ are present both in $Γ_i^⊥$ and in $Δ_i,A_i$, hence
    the atoms in $F$ are present both in $Γ^⊥$ and in $Δ$.
    The second case is similar except that we get $F:=F_1⅋F_2$.
  \item If $π$ ends with a $⅋$ rule, it has the shape
    \[
      \begin{prooftree}
        \Hypo{ &⊢ Γ', A, B, Δ' }
        \Infer1[$⅋$]{ &⊢ Γ', A⅋B, Δ' }
      \end{prooftree}
      \qquad\text{with}\qquad
      \begin{aligned}
        & Γ = Γ', \quad Δ = A⅋B, Δ' &&\text{or} \\
        & Γ = Γ', A⅋B, \quad Δ = Δ'
      \end{aligned}
    \]
    In the first case we get a formula $F$ and proofs of $⊢Γ',F$ and
    $⊢F^⊥,A,B,Δ'$ by induction hypothesis, and from the second one we
    immediately deduce a proof of $⊢F^⊥,A⅋B,Δ'$, so the same $F$ is appropriate.
    The second case is similar.
    The constraint on atoms is immediately satisfied.
  \item If $π$ ends with a dereliction, weakening or contraction rule, we get
    the expected formula immediately by induction hypothesis on the premiss.
  \item If $π$ ends with a promotion, it has the shape
    \[
      \begin{prooftree}
        \Hypo{ &⊢ \wn Γ', A, \wn Δ' }
        \Infer1[$\oc$]{ &⊢ \wn Γ', \oc A, \wn Δ' }
      \end{prooftree}
      \qquad\text{with}\qquad
      \begin{aligned}
        & Γ = \wn Γ', \quad Δ = \oc A, \wn Δ' &&\text{or} \\
        & Γ = \wn Γ', \oc A, \quad Δ = \wn Δ'
      \end{aligned}
    \]
    In the first case we get a formula $F$ and proofs of $⊢\wn Γ,F$ and
    $⊢F^⊥,A,\wn Δ'$, then by dereliction and promotion we get $⊢\wn Γ,\oc F$ and
    $⊢\wn F^⊥,\oc A,\wn Δ'$ (promotion on $A$) so $\oc F$ is appropriate.
    In the second case, similarly, we get $\wn F$ as the intermediate formula.
    \qedhere
  \end{itemize}
\end{proof}

\subsection{Typing and structural congruence (lemma \ref{lemma:congruence})}
\label{app:congruence}

\begin{proof}
  Thanks to lemma \ref{lemma:subtyping-commutation}, it is enough to
  consider typing derivations where \rulename{sub} rules only occur right above \rulename{new}
  rules (since no structural congruence rule involves inputs).

  For associativity, commutativity and neutrality in parallel composition, the
  associated properties for $⅋$ and $⊥$ are easily provable in multiplicative
  linear logic.

  For scope extrusion, consider a typed term $P\para\newk{x}{A}{k}Q$ where $x$
  does not occur in $P$.
  The typing derivation has the following shape:
  \begin{prooftree*}
    \Hypo{ E ⊢ P }
    \Hypo{ [x]_A^k ⊗ F &⊢ Q }
    \Infer1[new$k$]{ F &⊢ \newk{x}{A}{k}Q }
    \Infer2[para]{ E ⅋ F ⊢ P \para \newk{x}{A}{k}Q }
  \end{prooftree*}
  where $x$ does not occur in $E$ (since it does not occur in $P$) nor in $F$
  (by the side condition in \rulename{new$k$}).
  Then we can write the following derivation:
  \begin{prooftree*}
    \Hypo{ E ⊢ P }
    \Hypo{ [x]_A^k ⊗ F ⊢ Q }
    \Infer2[para]{ E ⅋ ([x]_A^k ⊗ F) ⊢ P \para Q }
    \Infer1[sub]{ [x]_A^k ⊗ (E ⅋ F) ⊢ P \para Q }
    \Infer1[new]{ E ⅋ F ⊢ \newk{x}{A}{k}(P \para Q) }
  \end{prooftree*}
  where the subtyping is easily proved in MLL.
  For the reverse rule, the typing of a term $\newk{x}{A}{k}(P\para Q)$ has the
  following shape:
  \begin{prooftree*}
    \Hypo{ E ⊢ P }
    \Hypo{ F ⊢ Q }
    \Infer2[para]{ E ⅋ F ⊢ P \para Q }
    \Infer1[sub]{ [x]_A^k ⊗ G &⊢ P \para Q }
    \Infer1[new]{ G &⊢ \newk{x}{A}{k}(P \para Q) }
  \end{prooftree*}
  The subtyping judgement is a proof in MELL of $⊢([x]_A^k)^⊥⅋G^⊥,E⅋F$, which
  is equivalent to $⊢([x]_A^k)^⊥,G^⊥,E,F$.
  By lemma \ref{lemma:interpolation}, we can deduce that there exists a MELL
  formula $H$ such that $⊢G^⊥,E,H$ and $⊢H^⊥,([x]_A^k)^⊥,F$ are provable and
  the literals in $H$ occur both in $G,E^⊥$ and in $([x]_A^k)^⊥,F$.
  By hypothesis $x$ does not occur in $P$ so it does not occur in $E$, by the
  side-condition on \rulename{new$k$} it does not occur in $G$ either,
  therefore $x$ does not occur in $H$.
  Therefore the literals in $H$ occur in $F$, which proves that $H$ only has
  positive literals, so it is an environment type.
  The proofs of $⊢G^⊥,E,H$ and $⊢H^⊥,([x]_A^k)^⊥,F$ induce subtypings $G≤E⅋H$
  and $[x]_A^k⊗H≤F$ so we can conclude this case by the following typing:
  \begin{prooftree*}
    \Hypo{ E ⊢ P }
    \Hypo{ F &⊢ Q }
    \Infer1[sub]{ [x]_A^k ⊗ H &⊢ Q }
    \Infer1[new]{ H &⊢ \newk{x}{A}{k}Q }
    \Infer2[para]{ E ⅋ H ⊢ P \para \newk{x}{A}{k}Q }
    \Infer1[sub]{ G ⊢ P \para \newk{x}{A}{k}Q }
  \end{prooftree*}

  For commutation of restrictions, a typed term $\newk{x}{A}{k}\newk{y}{B}{ℓ}P$ must have a
  derivation of the following shape:
  \begin{prooftree*}
    \Hypo{ [y]_B^ℓ ⊗ E &⊢ P }
    \Infer1[new]{ E &⊢ \newk{y}{B}{ℓ}P }
    \Infer1[sub]{ [x]_A^k ⊗ F &⊢ \newk{y}{B}{ℓ}P }
    \Infer1[new]{ F &⊢ \newk{x}{A}{k}\newk{y}{B}{ℓ}P }
  \end{prooftree*}
  From $[x]_A^k⊗F≤E$ we deduce $[x]_A^k⊗[y]_B^ℓ⊗F≤[y]_B^ℓ⊗E$, so we
  have the following typing:
  \begin{prooftree*}
    \Hypo{ [y]_B^ℓ ⊗ E &⊢ P }
    \Infer1[sub]{ [x]_A^k ⊗ [y]_B^ℓ ⊗ F &⊢ P }
    \Infer1[new]{ [y]_B^ℓ ⊗ F &⊢ \newk{x}{A}{k}P }
    \Infer1[new]{ F &⊢ \newk{y}{B}{ℓ}\newk{x}{A}{k}P }
  \end{prooftree*}
  which validates the case of commutation.
\end{proof}

\subsection{Subject reduction (theorem \ref{thm:subject-reduction})}
\label{app:subject-reduction}

\begin{proof}
  Thanks to lemma \ref{lemma:congruence}, we can reason up to structural
  congruence.
  For an interaction step between linear actions, we have
  $\newl{u}{A}(\outf{u}{¤v}.P\para\inb{u}{¤x}.Q)→P\para Q[¤v/¤x]$.
  The left-hand side is typed as follows:
  \begin{prooftree*}
    \Hypo{ E &⊢ P }
    \Infer1[out]{ u:↑A ⊗ (¤v:A ⅋ E) &⊢ \outf{u}{¤v}.P }
    \Hypo{ ¤x:A ⊗ F &⊢ Q }
    \Infer1[in]{ u:↓A ⊗ F &⊢ \inb{u}{¤x}.Q }
    \Infer2[para]{ (u:↑A⊗(¤v:A⅋E))⅋(u:↓A⊗F)
      ⊢ \outf{u}{¤v}.P\para\inb{u}{¤x}.Q }
    \Infer1[sub]{ [u]_A^1 ⊗ H ⊢ \outf{u}{¤v}.P\para\inb{u}{¤x}.Q }
    \Infer1[new1]{ H ⊢ \newl{u}{A}(\outf{u}{¤v}.P\para\inb{u}{¤x}.Q) }
  \end{prooftree*}
  with the hypothesis that no name in $¤x$ occurs in $F$.
  The natural typing for the reduct is obtained as follows:
  \begin{prooftree*}
    \Hypo{ E ⊢ P }
    \Hypo{ ¤v:A ⊗ F ⊢ Q[¤v/¤x] }
    \Infer2[para]{ E ⅋ (¤v:A ⊗ F) ⊢ P\para Q[¤v/¤x] }
  \end{prooftree*}
  By lemma \ref{lemma:mell-subst}, in the subtyping
  $[u]_A^1⊗H≤(u:↑A⊗(¤v:A⅋E))⅋(u:↓A⊗F)$ we can replace the atomic formula
  $u:↓A$ by $¤v:A$ and the atomic formula $u:↑A$ by $(¤v:A)^⊥$, then we get a
  proof of
  \[
    ⊢ ((¤v:A)^⊥ ⊗ ¤v:A) ⅋ H^⊥, ((¤v:A)^⊥⊗ (¤v:A ⅋ E)) ⅋ (¤v:A ⊗ F)
  \]
  The sequents
  \begin{gather*}
    ⊢ H^⊥, (¤v:A⅋(¤v:A)^⊥)⊗H \\
    ⊢ (¤v:A ⅋ ((¤v:A)^⊥ ⊗ E^⊥)) ⊗ ((¤v:A)^⊥ ⅋ F^⊥), E ⅋ (¤v:A ⊗ F)
  \end{gather*}
  are easily provable in MLL so by the cut rule we get a proof of
  $⊢H^⊥,(¤v:A⊗F)⅋E$ by which we can conclude with the typing of the reduct:
  \begin{prooftree*}
    \Hypo{ E ⊢ P }
    \Hypo{ ¤v:A ⊗ F ⊢ Q[¤v/¤x] }
    \Infer2[para]{ E ⅋ (¤v:A ⊗ F) ⊢ P\para Q[¤v/¤x] }
    \Infer1[sub]{ H ⊢ P\para Q[¤v/¤x] }
  \end{prooftree*}

  For an interaction step involving a replicated input, we have
  \[
    \newr{u}{A}(\outf{u}{¤v}.P\para\inr{u}{¤x}.Q\para R)
    → \newr{u}{A}(P\para Q[¤v/¤x]\para\inr{u}{¤x}.Q\para R) .
  \]
  The left-hand side is typed as follows:
  \begin{prooftree*}
    \Hypo{ E &⊢ P }
    \Infer1[out]{ u:↑A ⊗ (¤v:A ⅋ E) &⊢ \outf{u}{¤v}.P }
    \Hypo{ ¤x:A ⊗ F^! &⊢ Q }
    \Infer1[in!]{ \wn u:↓A ⊗ F^! &⊢ \inr{u}{¤x}.Q }
    \Hypo{ G ⊢ R }
    \Infer3[para]{ (u:↑A ⊗ (¤v:A ⅋ E)) ⅋ (\wn u:↓A ⊗ F^!) ⅋ G
      &⊢ \outf{u}{¤v}.P \para \inr{u}{¤x}.Q \para R }
    \Infer1[sub]{ [u]_A^ω ⊗ H &⊢ \outf{u}{¤v}.P \para \inr{u}{¤x}.Q \para R }
    \Infer1[newω]{ H &⊢ \newr{u}{A}(\outf{u}{¤v}.P \para \inr{u}{¤x}.Q \para R) }
  \end{prooftree*}
  Then we can deduce the following typing for the reduct without $\newr{u}{A}$:
  \begin{prooftree*}
    \Hypo{ E ⊢ P }
    \Hypo{ ¤v:A ⊗ F^! ⊢ Q[¤v/¤x] }
    \Hypo{ ¤x:A ⊗ F^! &⊢ Q }
    \Infer1[rep]{ \wn u:↓A ⊗ F^! &⊢ \inr{u}{¤x}.Q }
    \Hypo{ G ⊢ R }
    \Infer4[para]{ E ⅋ (¤v:A ⊗ F^!) ⅋ (\wn u:↓A ⊗ F^!) ⅋ G
      ⊢ P \para Q[¤v/¤x] \para \inr{u}{¤x}.Q \para R }
  \end{prooftree*}
  The instance of \rulename{sub} in the first typing uses a proof of
  \[
    ⊢ (\wn(u:↑A)^⊥ ⊗ \oc(u:↓A)^⊥) ⅋ H^⊥,
    (u:↑A ⊗ (¤v:A ⅋ E)) ⅋ (\wn u:↓A ⊗ F^!) ⅋ G
  \]
  Let $O:=u:↑A$ (output), $I:=u:↓A$ (input) and $V:=¤v:A$ (value).
  The sequent above is
  \[
    ⊢ (\wn O^⊥ ⊗ \oc I^⊥) ⅋ H^⊥,
    (O ⊗ (V ⅋ E)) ⅋ (\wn I ⊗ F^!) ⅋ G
  \]
  Because of the side condition in the rule \rulename{newω}
  the name $u$ does not occur in $H$ so there is no other occurrence of the
  literal $O^⊥$ in the above sequent, hence the linear atom $O$ can
  only be introduced as follows, up to permutations of rules:
  \begin{prooftree*}
    \Infer0[ax]{ ⊢ O^⊥, O }
    \Infer1[$\wn$]{ ⊢ \wn O^⊥, O }
  \end{prooftree*}
  If we replace this with
  \begin{prooftree*}
    \Infer0[ax]{ ⊢ V^⊥, V }
    \Infer1[w]{ ⊢ \wn O^⊥, V^⊥, V }
  \end{prooftree*}
  and we introduce a $⅋$ between $\wn O^⊥$ and $V^⊥$ just before
  $\wn O^⊥$ is involved in a tensor rule, we replace $O$ with $V$ in the
  proof above and we get a proof of
  \[
    ⊢ ((\wn O^⊥ ⅋ V^⊥) ⊗ \oc I^⊥) ⅋ H^⊥,
    (V ⊗ (V ⅋ E)) ⅋ (\wn I ⊗ F^!) ⅋ G
  \]
  The subformula $\wn I$ is necessarily introduced by a (possibly
  $η$-expanded) axiom rule that introduces $\oc I^⊥$, besides the latter
  only occurs once so $\wn I$ is only introduced once and thus is not
  involved in any contraction (except possibly with formulas introduced by
  weakening, but this case can be eliminated), so if we replace this axiom by
  an axiom on any formula and get another valid proof.
  Using the formula $V^⊥⅋\wn I$ we get
  \[
    ⊢ ((\wn O^⊥ ⅋ V^⊥) ⊗ (V ⊗ \oc I^⊥)) ⅋ H^⊥,
    (V ⊗ (V ⅋ E)) ⅋ ((V^⊥ ⅋ \wn I) ⊗ F^!) ⅋ G
  \]
  Composing this with the following proofs:
  \begin{prooftree*}
    \Infer0[ax]{ ⊢ \wn O^⊥, \oc O }
    \Infer0[ax]{ ⊢ V, V^⊥ }
    \Infer2[$⊗$]{ ⊢ \wn O^⊥, \oc O ⊗ V, V^⊥ }
    \Infer0[ax]{ ⊢ \oc I^⊥, \wn I }
    \Infer2[$⊗$]{ ⊢ \wn O^⊥ ⊗ \oc I^⊥,
      \oc O ⊗ V, V^⊥, \wn I }
    \Infer[double]1[$⅋$]{ ⊢ \wn O^⊥ ⊗ \oc I^⊥,
      (\oc O ⊗ V) ⅋ (V^⊥ ⅋ \wn I) }
    \Infer0[ax]{ ⊢ H^⊥, H }
    \Infer2[$⊗$]{ ⊢ \wn O^⊥ ⊗ \oc I^⊥, H^⊥,
      ((\oc O ⊗ V) ⅋ (V^⊥ ⅋ \wn I)) ⊗ H }
    \Infer1[$⅋$]{ ⊢ (\wn O^⊥ ⊗ \oc I^⊥) ⅋ H^⊥,
      ((\oc O ⊗ V) ⅋ (V^⊥ ⅋ \wn I)) ⊗ H }
  \end{prooftree*}
  and
  \begin{prooftree*}
    \Infer0[ax,ax,$⊗$]{ ⊢ V^⊥, V ⊗ E^⊥, E }
    \Infer1[$⅋$]{ ⊢ V^⊥ ⅋ (V ⊗ E^⊥), E }
    \Infer0[ax]{ ⊢ V^⊥, V }
    \Infer0[ax]{ ⊢ (F^!)^⊥, F^! }
    \Infer2[$⊗$]{ ⊢ V^⊥, (F^!)^⊥, V ⊗ F^! }
    \Infer0[ax]{ ⊢ \oc I^⊥, \wn I }
    \Infer0[ax]{ ⊢ (F^!)^⊥, F^! }
    \Infer2[$⊗$]{ ⊢ \oc I^⊥, (F^!)^⊥, \wn I ⊗ F^! }
    \Infer2[$⊗$]{ ⊢ V^⊥ ⊗ \oc I^⊥, (F^!)^⊥, (F^!)^⊥, V ⊗ F^!, \wn I ⊗ F^! }
    \Infer1[c]{ ⊢ V^⊥ ⊗ \oc I^⊥, (F^!)^⊥, V ⊗ F^!, \wn I ⊗ F^! }
    \Infer1[$⅋$]{ ⊢ (V^⊥ ⊗ \oc I^⊥) ⅋ (F^!)^⊥, V ⊗ F^!, \wn I ⊗ F^! }
    \Infer[separation=-4em]2[$⊗$]{ ⊢
      (V^⊥ ⅋ (V ⊗ E^⊥)) ⊗ ((V^⊥ ⊗ \oc I^⊥) ⅋ (F^!)^⊥),
      E, V ⊗ F^!, \wn I ⊗ F^! }
    \Infer0[ax]{ ⊢ G^⊥, G }
    \Infer[separation=-4em]2[$⊗$]{ ⊢
      (V^⊥ ⅋ (V ⊗ E^⊥)) ⊗ ((V^⊥ ⊗ \oc I^⊥) ⅋ (F^!)^⊥) ⊗ G^⊥,
      E, V ⊗ F^!, \wn I ⊗ F^!, G }
    \Infer[double]1[$⅋$]{ ⊢
      (V^⊥ ⅋ (V ⊗ E^⊥)) ⊗ ((V^⊥ ⊗ \oc I^⊥) ⅋ (F^!)^⊥) ⊗ G^⊥,
      E ⅋ (V ⊗ F^!) ⅋ (\wn I ⊗ F^!) ⅋ G }
  \end{prooftree*}
  we get
  \[
    ⊢ (\wn O^⊥ ⊗ \oc I^⊥) ⅋ H^⊥, E ⅋ (V ⊗ F^!) ⅋ (\wn I ⊗ F^!) ⅋ G
  \]
  that is
  \[
    ⊢ (\wn (u:↑A)^⊥ ⊗ \oc(u:↓A)^⊥) ⅋ H^⊥,
    E ⅋ (¤v:A ⊗ F^!) ⅋ (\wn(u:↓A) ⊗ F^!) ⅋ G
  \]
  hence we have
  \begin{prooftree*}
    \Hypo{ E ⊢ P }
    \Hypo{ ¤v:A ⊗ F^! ⊢ Q[¤v/¤x] }
    \Hypo{ ¤x:A ⊗ F^! &⊢ Q }
    \Infer1[rep]{ \wn u:↓A ⊗ F^! &⊢ \inr{u}{¤x}.Q }
    \Hypo{ G ⊢ R }
    \Infer4[para]{ E ⅋ (¤v:A ⊗ F^!) ⅋ (\wn u:↓A ⊗ F^!) ⅋ G
      ⊢ P \para Q[¤v/¤x] \para \inr{u}{¤x}.Q \para R }
    \Infer1[sub]{ (\wn (u:↑A) ⅋ \oc(u:↓A)) ⊗ H
      ⊢ P \para Q[¤v/¤x] \para \inr{u}{¤x}.Q \para R }
    \Infer1[newω]{ H
      ⊢ \newr{u}{A}(P \para Q[¤v/¤x] \para \inr{u}{¤x}.Q \para R) }
  \end{prooftree*}
  which concludes the proof.
\end{proof}

\end{document}